\documentclass[conference,10pt]{IEEEtran}
\usepackage{amsmath, mathrsfs, mathtools}
\usepackage{amsfonts}
\usepackage{amssymb}
\usepackage{color}
\usepackage{multirow}
\usepackage{stmaryrd}

\usepackage{yfonts}

\allowdisplaybreaks

\usepackage{caption}
\usepackage{subcaption}

\usepackage{float}

\usepackage{amsfonts}
\usepackage{amssymb}
\usepackage{color}
\usepackage{multirow}
\usepackage{graphicx}
\usepackage{tcolorbox}
\usepackage{mathrsfs}

\newcommand{\subparagraph}{}
\usepackage[compact]{titlesec}
\titlespacing*{\section}{15pt}{1.2\baselineskip}{0.9\baselineskip}

\allowdisplaybreaks

\newcommand{\myhash}{%
	{\settoheight{\dimen0}{C}\kern-.05em\, \resizebox{!}{\dimen0}{\raisebox{\depth}{\#}}}}

\usepackage[numbers,sort&compress]{natbib}

\newcommand{\Sigmay}{{\Sigmay}_{\yv}}

\usepackage{multicol}

\usepackage{algorithm}
\usepackage{algpseudocode}
\usepackage{pifont}
\usepackage{varwidth}


\usepackage{mathbbol}

\def\mindex#1{\index{#1}}

\def\sq{\hbox{\rlap{$\sqcap$}$\sqcup$}}
\def\qed{\ifmmode\sq\else{\unskip\nobreak\hfil
\penalty50\hskip1em\null\nobreak\hfil\sq
\parfillskip=0pt\finalhyphendemerits=0\endgraf}\fi\medskip}

\long\def\defbox#1{\framebox[.9\hsize][c]{\parbox{.85\hsize}{%
\parindent=0pt
\baselineskip=12pt plus .1pt      
\parskip=6pt plus 1.5pt minus 1pt 
 #1}}}

\long\def\beginbox#1\endbox{\subsection*{}%
\hbox{\hspace{.05\hsize}\defbox{\medskip#1\bigskip}}%
\subsection*{}}

\def\endbox{}

\newsavebox{\junk}
\savebox{\junk}[1.6mm]{\hbox{$|\!|\!|$}}

\def\det{{\mathop{\rm det}}}

\def\Re{\field{R}}

\def\bC{{\mathbb C}}

\def\bE{{\mathbb E}}

\def\bR{{\mathbb R}}

\def\scrC{{\mathscr{C}}}

\def\sfH{{\sf H}}

\def\bfmath#1{{\mathchoice{\mbox{\boldmath$#1$}}%
{\mbox{\boldmath$#1$}}%
{\mbox{\boldmath$\scriptstyle#1$}}%
{\mbox{\boldmath$\scriptscriptstyle#1$}}}}

\def\bfmY{\bfmath{Y}}

\def\bfmhhaY{\bfmath{\hhaY}} 
\def\bfmhhaY{\hbox to 0pt{$\widehat{\bfmY}$\hss}\widehat{\phantom{\raise 1.25pt\hbox{$\bfmY$}}}}

\def\til={{\widetilde =}}

 \def\FRAC#1#2#3{\genfrac{}{}{}{#1}{#2}{#3}}

\def\ddtp{{\mathchoice{\FRAC{1}{d^{\hbox to 2pt{\rm\tiny +\hss}}}{dt}}%
{\FRAC{1}{d^{\hbox to 2pt{\rm\tiny +\hss}}}{dt}}%
{\FRAC{3}{d^{\hbox to 2pt{\rm\tiny +\hss}}}{dt}}%
{\FRAC{3}{d^{\hbox to 2pt{\rm\tiny +\hss}}}{dt}}}}

\def\average#1,#2,{{1\over #2} \sum_{#1}^{#2}}

\def\eye(#1){{\bf(#1)}\quad}

\newtheorem{theorem}{{\bf Theorem}}

\newtheorem{definition}{{\bf Definition}}

\newtheorem{lemma}{{\bf Lemma}}

\def\eq#1/{(\ref{e:#1})}

\newcommand{\beqn}[1]{\notes{#1}%
\begin{eqnarray} \elabel{#1}}

\newcommand{\eeqn}{\end{eqnarray} }

\newcommand{\beq}[1]{\notes{#1}%
\begin{equation}\elabel{#1}}

\newcommand{\eeq}{\end{equation}}

\def\bdes{\begin{description}}
\def\edes{\end{description}}

\newcounter{rmnum}

\newcounter{anum}

{\end{list}}

\def\ass(#1:#2){(#1\ref{#1:#2})}

\def\ritem#1{
\item[{\sf \ass(\current_model:#1)}]
}

\newenvironment{recall-ass}[1]{%
\begin{description}
\def\current_model{#1}}{
\end{description}
}

\usepackage{float}
\usepackage{afterpage}
\usepackage{balance}

\usepackage{geometry}
\def\herm{{\sfH}}

\newcommand{\normd}[1]{{\left\vert\kern-0.25ex\left\vert\kern-0.25ex\left\vert #1 
		\right\vert\kern-0.25ex\right\vert\kern-0.25ex\right\vert}}
\newcommand\blfootnote[1]{%
	\begingroup
	\renewcommand\thefootnote{}\footnote{#1}%
	\addtocounter{footnote}{-1}%
	\endgroup
}

\long\def\comment#1{}

\newcommand{\av}{{\bf a}}

\newcommand{\dv}{{\bf d}}

\newcommand{\hv}{{\bf h}}

\newcommand{\rv}{{\bf r}}
\newcommand{\sv}{{\bf s}}

\newcommand{\uv}{{\bf u}}

\newcommand{\xv}{{\bf x}}
\newcommand{\yv}{{\bf y}}

\newcommand{\Am}{{\bf A}}
\newcommand{\Bm}{{\bf B}}

\newcommand{\Dm}{{\bf D}}

\newcommand{\Fm}{{\bf F}}

\newcommand{\Hm}{{\bf H}}

\newcommand{\Sm}{{\bf S}}
\newcommand{\Tm}{{\bf T}}
\newcommand{\Um}{{\bf U}}

\newcommand{\Vm}{{\bf V}}

\newcommand{\Bc}{{\cal B}}
\newcommand{\Cc}{{\cal C}}

\newcommand{\Nc}{{\cal N}}

\newcommand{\Pc}{{\cal P}}

\newcommand{\Uc}{{\cal U}}

\newcommand{\lambdav}{\hbox{\boldmath$\lambda$}}

\newcommand{\xiv}{\hbox{\boldmath$\xi$}}
\newcommand{\sigmav}{\hbox{\boldmath$\sigma$}}

\newcommand{\Lambdam}{\hbox{\boldmath$\Lambda$}}

\newcommand{\Sigmam}{\hbox{\boldmath$\Sigma$}}

\newcommand{\Xim}{\hbox{\boldmath$\Xi$}}

\renewcommand{\det}{{\hbox{det}}}

\renewcommand{\Re}{{\rm Re}}

\newcommand{\transp}{{\sf T}}
\renewcommand{\vec}{{\rm vec}}

\title{Joint Approximate Covariance Diagonalization\\ with Applications in MIMO Virtual Beam Design}
\author{\IEEEauthorblockN{Mahdi Barzegar Khalilsarai, Saeid Haghighatshoar\IEEEauthorrefmark{1}, and Giuseppe Caire}
	\vspace{-4mm}\\
	Communications and Information Theory Group, Technische Universit\"{a}t Berlin\\
	Emails: $\{$m.barzegarkhalilsarai, saeid.haghighatshoar, caire$\}$@tu-berlin.de}

\begin{document}

\maketitle

\def\ful{f_\text{ul}}
\def\fdl{f_\text{dl}}
\def\asfc{\scrC}
\def\asful{\scrC_\text{ul}}
\def\asfdl{\scrC_\text{dl}}
\def\fproxy{f_{\sf proxy}}
\begin{abstract}
We study the problem of maximum-likelihood (ML) estimation of an approximate common eigenstructure, i.e. an approximate common eigenvectors set (CES), for an ensemble of covariance matrices given a collection of their associated i.i.d vector realizations. This problem has a direct application in multi-user MIMO communications, where the base station (BS) has access to instantaneous user channel vectors through pilot transmission and attempts to perform joint multi-user Downlink (DL) precoding. It is widely accepted that an efficient implementation of this task hinges upon an appropriate design of a set of common ``virtual beams", that captures the common eigenstructure among the user channel covariances. In this paper, we propose a novel method for obtaining this common eigenstructure by casting it as an ML estimation problem. We prove that in the special case where the covariances are jointly diagonalizable, the global optimal solution of the proposed ML problem coincides with the common eigenstructure. Then we propose a projected gradient descent (PGD) method to solve the ML optimization problem over the manifold of unitary matrices and prove its convergence to a stationary point. Through exhaustive simulations, we illustrate that in the case of jointly diagonalizable covariances, our proposed method converges to the exact CES. Also, in the general case where the covariances are not jointly diagonalizable, it yields a solution that \textit{approximately} diagonalizes all covariances. Besides, the empirical results show that our proposed method outperforms the well-known \textit{joint approximate diagonalization of eigenmatrices} (JADE) method in the literature.
\end{abstract}

\section{Introduction}
\blfootnote{\IEEEauthorrefmark{1}Saeid Haghighatshoar is currently with the Swiss Center for Electronics and Microtechnology (CSEM), however his contribution to this work was made while he was with the CommIT group (saeid.haghighatshoar@csem.ch).}
Massive multiple-input multiple-output (MIMO) communication systems have recently received considerable attention by promising unprecedented data rates, network coverage and link reliability \cite{larsson2014massive,bjornson2016massive}. This technology makes use of a large number ($M\gg 1$) of antennas at the base station (BS), enabling the spatial multiplexing of several data streams to a number of users $K$ that share the same time-frequency transmission resource \cite{yang2013performance}. Massive MIMO can be implemented for both narrow-band and wide-band signaling and is compatible with new trends such as millimeter wave (mmWave) communications \cite{heath2016overview}. 

In order to precode effectively the Downlink (DL) data streams to the users, massive MIMO requires some level of channel state information at the transmitter side (CSIT). This can be easily obtained in TDD systems owing to channel reciprocity and calibrated RF front-ends \cite{marzetta2010noncooperative,shepard2012argos}. 
However, in the case where the computation of the precoder is too complex for real-time operations, or in the FDD case, where obtaining accurate and timely CSIT is very difficult, 
several schemes have been proposed in order to exploit the statistical structure of the massive MIMO channels, and in particular to decompose them into ``virtual beam directions", such that 
the precoding and channel estimation can be performed in a virtual beam domain of reduced dimension \cite{adhikary2013joint,haghighatshoar2016massive,haghighatshoar2018low,khalilsarai2018fdd}. In particular, in many typical propagation environments the number of scatterers is much less than the number of array elements, the BS is connected to a user through a narrow angular aperture. Therefore, by a careful choice of the beam-space, the BS can use a variety of sparse signal processing tools to efficiently perform its tasks \cite{haghighatshoar2017massive,khalilsarai2018achieve,nguyen2013compressive,wunder2018hierarchical}. 

Defining the array response vector as $\av (\xiv) = [ e^{j\tfrac{2\pi}{\lambda}\langle \rv_1 , \xiv \rangle},\ldots,e^{j\tfrac{2\pi}{\lambda}\langle \rv_M , \xiv \rangle}, ]^\transp$, the MIMO channel covariance can be written as \cite{khalilsarai2018fdd}
\begin{equation}\label{eq:ch_cov_0}
\widetilde{\Sigmam}  = \bE \left[ \hv \hv^\herm \right] = \int_{\Xim} \gamma (\xiv) \av (\xiv) \av (\xiv)^\herm d\xiv,  
\end{equation}
 where $\lambda$ is the wavelength, $\{\rv_m \}_{m=1}^M$ denotes antenna locations, $\xiv \in \Xim$ is the angle of arrival (AoA) and $\gamma (\xiv)$ is the angular power spread function. 
The beam-space for a user is characterized by a set of vectors, called \textit{virtual beams} that span the column space of $\widetilde{\Sigmam}$. The ``best" choice for such vectors is given by the eigenvectors of $\widetilde{\Sigmam}$ denoted as columns of the matrix $\Um = [\uv_1,\ldots,\uv_M]$, which also comprise the basis for the Karhunen-L{\`o}eve (KL) expansion of $\hv$ as
$
\hv = \sum_{m=1}^{M} w_m \uv_m,
$
where $w_m \in \bC, \, m=1,\ldots,M$ are uncorrelated random variables. The KL expansion provides the best $k$-term approximation for $\hv$ and has, exclusive to itself, the property of decorrelating $\hv$ into orthonormal vectors $\{ \uv_m \}_{m=1}^M$ with uncorrelated coefficients $\{w_m \}_{m=1}^M$ \cite{unser2014introduction}. For Gaussian channels, widely assumed in wireless multipath channel models, this property is even more pronounced as the coefficients become statistically independent. The variance of a coefficient $\sigma_m^2 =\bE [ |w_m|^2 ] $ represents the channel ``power" along virtual beam $\uv_m$ and reveals the channel angular sparsity, such that if the power along a beam is less than a threshold ($\sigma_m^2 < \epsilon$), one can assume that the BS receives almost no energy from that beam.

While obtaining the set of virtual beams is easy for a single user, it is not straightforward when one deals with multiple users, since in general the users do not share the same set of eigenvectors. On the other hand, having a suitable set of common virtual beams (CVB) is highly desirable for DL channel training and multi-user precoding \cite{sayeed2013beamspace,khalilsarai2018fdd}. Most of the works in the literature \textit{assume} the existence of a common basis of orthogonal beam directions that (approximately) diagonalizes all user covariances and represents a sort of common eigenvectors set (CES) \cite{heath2016overview}. For the simplest case of a \textit{uniform linear array} (ULA), it is typically taken for granted that the CVB is given by a uniform sampling of the array response vector, which in turn is equivalent to the DFT basis. The implicit (and seldom mentioned) reason for this assumption comes from an asymptotic ($M\to \infty$) equivalence between Toeplitz and circulant matrices according to the Szeg{\"o} theorem (see \cite{adhikary2013joint} and references therein). Since the covariance of the channel associated with a ULA is Toeplitz, and since circulant matrices are diagonalized by the DFT matrix, it holds that DFT columns form an approximate set of eigenvectors for any ULA covariance. As this theoretic result holds for large ULAs, the choice of a suitable set of virtual beams for small or moderate-sized arrays is unclear. Besides, for array geometries other than the ULA, such as circular, planar and cylindrical arrays, no universal virtual beam design is known. 
%

This paper provides a solution for the problem above. Given random channel realizations $\hv_k^{(1)},\, \hv_k^{(2)},\, \ldots, \, \hv_k^{(N)}$ of a set of users $k=1,\ldots,K$ and without knowing the array geometry we propose a rigorous method for obtaining an orthogonal CVB set. Note that this problem is more general compared to the case in which we have user covariances $\widetilde{\Sigmam}_k,\, k=1,\ldots, K$. Our method is based on the maximum-likelihood (ML) estimation of the \textit{postulated} CES given random channel realizations. We formulate the ML problem as an optimization over the unitary manifold and proposed a \textit{projected gradient descent} (PGD) method to solve it. We prove the convergence of this algorithm to a stationary point of the likelihood function with arbitrary initialization. We also show that, with jointly diagonalizable covariances, the CES coincides with the global maximizer of the likelihood function. Finally, we compare our method to the classic joint approximate diagonalization of eigenmatrices (JADE) algorithm \cite{cardoso1993blind} to show its superior performance. 

\section{System Setup}
Consider the following scenario: we observe $N$ random realizations for each of the $K$ independent, stationary, zero-mean vector Gaussian processes as
\begin{equation}\label{eq:vector_observations}
\Hm_k  = \left[ \hv_k^{(1)},\ldots,\hv_k^{(N)} \right] \in \bC^{M\times N},\, k=1,\ldots,K. 
\end{equation} 
The covariance matrix of process $k$ is denoted by $\widetilde{\Sigmam}_k = \bE [\hv_k \hv_k^\herm  ]$. The random vector $\hv_k$ can represent, for instance, the channel vector of a user $k$ communicating with a base station (BS) equipped with $M$ antennas. The eigendecomposition of  $\widetilde{\Sigmam}_k$ is given as
\begin{equation}\label{eq:eigdecomp_true}
\widetilde{\Sigmam}_k = \Um_k \widetilde{\Lambdam}_k \Um_k^\herm, 
\end{equation}
where $\Um_k$ is the unitary matrix of eigenvectors ($\Um_k^\herm \Um_k=\mathbf{I}_M$) and $\widetilde{\Lambdam}_k$ is the diagonal matrix of eigenvalues. 

We are interested in obtaining a (approximate) common eigenstructure\footnote{We use the terms ``common eigenstructure", ``common eigenvectors set", and ``common virtual beams set" interchangeably.} among all covariances $\{ \widetilde{\Sigmam}_k \}_{k=1}^K$ given random samples $\{\Hm_k\}_{k=1}^K$. If the covariance matrices are jointly diagonalizable, i.e. if there exists a unitary matrix $\Um^c$ such that $\Um_1=\Um_2=\ldots=\Um_K=\Um^c$, then it is desirable to obtain $\Um^c$ as the common eigenstructure. If the covariances are \textit{not} jointly diagonalizable, then we want to obtain a unitary matrix $\Um^\star$ as the common eigenstructure that best diagonalizes the covariances.

 To have a systematic way of estimating the common eigenstructure, we first impose the joint diagonalizability criterion on the estimation model, in which each covariance is decomposed as
\begin{equation}\label{eq:decomp}
 \Sigmam_k = \Um \Lambdam_k \Um^\herm, 
\end{equation}
where $\Um=[\uv_1,\ldots,\uv_M]\in \bC^{M\times M}$ is a unitary matrix ($\Um^\herm \Um=\mathbf{I}_M$) representing the to-be-estimated common eigenstructure (or CVB set) and, assuming non-singular covariances for simplicity, $\Lambdam_k = \text{diag}(\lambdav_k)$ is a diagonal matrix with positive diagonal elements given in the vector $\lambdav_k$ for $k=1,\ldots,K$. Note the difference between the true covariance $\widetilde{\Sigmam}_k$ in \eqref{eq:eigdecomp_true} and the hypothetical covariance $\Sigmam_k$ in \eqref{eq:decomp}. Then, given random samples as in \eqref{eq:vector_observations}, we solve a maximum-likelihood optimization to estimate $\Um$. This approach can be seen as a \textit{deliberately mismatched} ML estimation method: the covariances do \textit{not} generally share the same set of eigenvectors, but we impose this property by adopting the model in \eqref{eq:decomp}. Notice that by imposing such model mismatch, we look for the common unitary matrix $\Um$ that best fits the sample data $\Hm_1,\ldots, \Hm_K$, where ``best fit" is to be interpreted in the Maximum Likelihood sense. 


Observing the random samples $\{ \Hm_k \}_{k=1}^K$, one can write the likelihood function as a function of the hypothetical covariance matrices according to
\begin{equation}\label{eq:ML_func}
\begin{aligned}
& p\left(\{\Hm_k \}_{k=1}^K |\{ \Sigmam_k\}_{k=1}^K \right) =\\
 &\prod_{k=1}^{K} \frac{1}{(2\pi)^{\frac{M}{2}} \det (\Sigmam_k)^{\frac{N}{2}}}\exp \left( -\frac{1}{2N}\text{trace}(\Hm_k^\herm \Sigmam_k^{-1} \Hm_k)\right)=\\
& \prod_{k=1}^{K} \frac{1}{(2\pi)^{\frac{M}{2}} \det (\Sigmam_k)^{\frac{N}{2}}}\exp \left( -\frac{1}{2}\text{trace}( \Sigmam_k^{-1}   \widehat{\Sigmam}_k )\right),
\end{aligned}
\end{equation}
where we have defined $\widehat{\Sigmam}_k:= \frac{1}{N} \Hm_k \Hm_k^\herm$, as the sample covariance of the $k$-th process. Taking the $-\log(\cdot)$ of the likelihood function, scaling it, omitting constant terms, and replacing $\Sigmam_k$ with $\Um \Lambdam_k \Um^\herm$ from \eqref{eq:decomp}, one can show that maximizing the likelihood function is equivalent to minimizing the following ML cost as a function of the parameters $\Um$ and $ \{ \lambdav_k \}_{k=1}^K$:
\begin{equation}\label{eq:cost_func}
\begin{aligned}
& \Cc \left(\Um , \{ \lambdav_k \}_{k=1}^K \right) =\\
&\scalebox{0.92}{$\sum_{k=1}^{K} \log\det\, (\Um \text{diag}(\lambdav_k) \Um^\herm) + \text{trace} \left( (\Um \text{diag}(\lambdav_k) \Um^\herm)^{-1} \widehat{\Sigmam}_k \right).$}
\end{aligned}
\end{equation}
 Since 
$\det\, (\Um \text{diag}(\lambdav_k) \Um^\herm) =\nolinebreak \det\, (\text{diag}(\lambdav_k) ),$
 and \scalebox{0.95}{$\text{trace} \left( (\Um \text{diag}(\lambdav_k) \Um^\herm)^{-1} \widehat{\Sigmam}_k \right) = \nolinebreak \text{trace} \left(  \text{diag}(\lambdav_k)^{-1} \Um^\herm \widehat{\Sigmam}_k\Um \right),$} we have
\begin{equation}
\Cc \left(\Um , \{ \lambdav_k \}_{k=1}^K \right) = \sum_{m,k} \log \lambdav_{k,m} + \frac{\uv_m^\herm \widehat{\Sigmam}_k \uv_m}{\lambdav_{k,m}},
\end{equation}
where $\lambdav_{k,m}> 0$ is the $m$-th element of $\lambdav_k$. The ML optimization problem then can be formulated as
\begin{equation}\label{eq:ML_formula_1}
\begin{aligned}
&\underset{\{ \uv_m \}_{m=1}^M ,\{\lambdav_k \}_{k=1}^K}{\text{minimize}} && \sum_{m,k} \log \lambdav_{k,m} + \frac{\uv_m^\herm \widehat{\Sigmam}_k \uv_m}{\lambdav_{k,m}}\\
&\hspace{6mm} \text{subject to} && \uv_m^\herm \uv_n = \delta_{m,n}
\end{aligned}
\end{equation}
For simplicity we assume that all sample covariance matrices are non-singular, such that $\uv_m^\herm \widehat{\Sigmam}_k \uv_m>0$ for any $\uv_m$. As a result, it is easy to show that, for given $\{ \uv_m\}_{m=1}^M$, the function $g(x)=\log x + \frac{\uv_m^\herm \widehat{\Sigmam}_k \uv_m}{x},\, x>0$ achieves its minimum at $x=\uv_m^\herm\widehat{ \Sigmam}_k \uv_m$. Therefore, we take the minimization in \eqref{eq:ML_formula_1} first with respect to $\lambdav_{k,m}, \, k\in [K],\, m\in [M]$ (for an integer $n$, we define $[n]:=\{1,\ldots,n \}$) and transform \eqref{eq:ML_formula_1} to 
\begin{equation}\label{eq:ML_formula_2}
\begin{aligned}
&\underset{\{ \uv_m \}_{m=1}^M }{\text{minimize}} && f(\Um)=\sum_{m,k} \log \left( \uv_m^\herm \widehat{\Sigmam}_k \uv_m\right) \\
& \text{subject to} && \uv_m^\herm \uv_n = \delta_{m,n}
\end{aligned}\tag{$P_{ ML}$}
\end{equation}
This presents an optimization problem over the manifold of unitary matrices $\Uc = \{ \Um \in \bC^{M\times M}:\, \Um^\herm \Um =\mathbf{I}_M \}$. To solve \ref{eq:ML_formula_2}, we propose a gradient projection method and show that it converges to a stationary point of the cost function $f(\Um)$.
\subsection{Intermezzo: Jointly Diagonalizable Covariances and Global Optimality of the CES}
In the discussion above, we said that for jointly diagonalizable covariances, we with to obtain the CES $\Um^c$ as a result of our estimation method. In order to show that the ML problem \eqref{eq:ML_formula_2} gives a reasonable framework to satisfy this property, here we prove that for jointly diagonalizable covariances and by assuming sample covariances to have converged ($\widehat{ \Sigmam}_k = \Sigmam_k$), the CES $\Um^c$ is in fact a global minimizer of $P_{ML}$. Let the set of covariances $\Sigmam_k,~k\in [K]$ to be decomposed as
\begin{equation}\label{eq:commut_cov}
\Sigmam_k = \Um^c \Lambdam_k \Um^{c\, \herm},
\end{equation} 
where $\Um^c \in \bC^{M\times M}$ denotes the CES. Assume $\widehat{ \Sigmam}_k = \Sigmam_k$, and consider the following definition. 
\begin{definition}[Majorization]
	For $\xv \in \bR^{M}$, define $\xv^{\downarrow}$ as the vector containing the elements of $\xv$ in descending order. Let $\yv \in \bR^M$ be another vector such that $\sum_{i=1}^{M}\xv_i =\sum_{i=1}^{M}\yv_i  $. We say $\xv$ \textit{majorizes} $\yv$ ($\xv \succ \yv$) iff 
	\[ \sum_{i=1}^{m}\xv^{\downarrow}_i \ge \sum_{i=1}^{m}\yv^{\downarrow}_i,  \]
	for all $m\in [M]$.
\end{definition}
We have the following theorem on global optimality of $\Um^c$.
\begin{theorem}\label{thm_global_opt}
	Let $\Sigmam_k,\, k=1,\ldots, K$ be a set of jointly diagonalizable covariance matrices as in \eqref{eq:commut_cov} and consider the optimization problem
	\begin{equation}\label{eq:opt_new}
	\underset{\Um}{\text{minimize}}\, f(\Um)=\sum_{m,k} \log \left( \uv_m^\herm \Sigmam_k \uv_m  \right) ~\text{s.t.}~ \uv_m^\herm \uv_n = \delta_{m,n}.
	\end{equation}
	Then $\Um^\star = \Um^c$ is a global solution of \eqref{eq:opt_new}.
\end{theorem}
\begin{proof}
For any unitary $\Um$, define the vector $\sigmav_k (\Um) \in \bR^M$ where $[\sigmav_k (\Um)]_{m} = \uv_m^\herm \Sigmam_k \uv_m$. In particular $\sigmav_k (\Um^c)$ is the vector of eigenvalues of $\Sigmam_k$. Using the properties of eigenvalue decomposition one can show $\sigmav_k (\Um^c) \succ \sigmav_k (\Um)$ for all $\Um \in \Uc$ and all $k\in [K]$.
	In addition, the function $h(\xv) = \sum_i \log (\xv_i)$ is Schur-concave \cite{peajcariaac1992convex} and therefore $\sum_m \log ([\sigmav_k (\Um^c)]_m) \le \sum_m \log ([\sigmav_k (\Um)]_m)$. Hence, $f (\Um^c)\le f (\Um)$ for all $\Um \in \Uc$, proving $\Um^c$ to be the global minimizer of $f(\Um)$ over $\Uc$.
\end{proof}

\section{ML via Projected Gradient Descent}
The projected gradient descent method (PGD) is a well-known iterative optimization algorithm \cite{bertsekas2015convex}. Starting from an initial point $\Um^{(0)}$, this method consists of the two following steps per iteration:
\begin{equation}\label{eq:PGD_step1}
\widetilde{\Um}^{(t)}  = \Um^{(t)} - \alpha_t \nabla f(\Um^{(t)})\tag{Gradient Step}
\end{equation}
\begin{equation}\label{eq:PGD_step2}
\Um^{(t+1)} = \Pc_{\Uc} (\widetilde{\Um}^{(t)}) \tag{Projection Step}
\end{equation}
where $\alpha_t>0$ is a step size, $\nabla f(\Um^{(t)})\in \bC^{M\times M}$ is the gradient of $f$ at $\Um^{(t)}$ and $\Pc_{\Uc}:\bC^{M\times M}\to \Uc $ denotes the orthogonal projection operator onto the set of unitary matrices. The explicit expression of this operator is given in \cite{manton2002optimization}; Nevertheless, we derive it here through the following lemma for the sake of completeness. 
\begin{lemma}
	Let $\Vm \in \bC^{M\times M}$ be a matrix with singular value decomposition $\Vm = \Sm \Dm \Tm^\herm$ where $\Sm$ and $\Tm$ are unitary matrices of left and right eigenvectors and $\Dm=\text{diag}(\dv)$ is non-negative diagonal. Then, the orthogonal projection of $\Vm$ onto the set of unitary matrices is given by $\Pc_{\Uc}(\Vm) =\Sm \Tm^\herm$.
\end{lemma}
\begin{proof}
	The orthogonal projection of $\Vm$ is given by the minimizer of $g(\Um)= \Vert \Vm -\Um\Vert_F^2$ over the set of unitary matrices, where $\Vert\cdot \Vert_F^2$ denotes Frobenius norm and $\Um^\herm\Um=\mathbf{I}_M$. We can write
	\begin{equation}\label{eq:fro}
	\begin{aligned}
	g(\Um)= \Vert \Vm -\Um\Vert_F^2 & = \Vert \Um\Vert_{\sf F}^2 + \Vert \Vm \Vert_{\sf F}^2 - 2\Re\{ \langle\Vm , \Um \rangle \}\\
	&  =M + \Vert \Vm \Vert_{\sf F}^2 - 2\Re\{ \langle\Vm , \Um \rangle \},
	\end{aligned}
	\end{equation}
	where the inner product is defined as $\langle\Vm , \Um \rangle  = \text{trace} (\Um^\herm \Vm) $ and we used the fact that $ \text{trace} (\Um \Um^\herm)=\text{trace}(\mathbf{I}_M)=M$. According to Von Neumann's trace inequality we have 
$| \langle\Vm , \Um \rangle  |=|\text{trace}\left( \Um^\herm\Vm \right)|\le \langle \sv_{\Um},\dv \rangle,  $\cite{mirsky1975trace}, where $\sv$ denotes the singular values vector of $\Um$. In the special case where $\Um$ is unitary, we have $\sv_{\Um} = [1,\ldots,1]^\transp$ and 
	$|\langle\Vm , \Um \rangle |\le  \langle \sv_{\Um},\dv \rangle = \sum_i \dv_i.$
	Now, using \eqref{eq:fro} we have  
$	g(\Um)\ge M +\Vert \Vm\Vert_{\sf F}^2 -2\sum_i \dv_i,$
where the right hand side of the inequality is independent of $\Um$. We show that the lower bound on $g(\Um)$ is achieved by $\Um^\star = \Sm \Tm^\herm$. This is seen by the fact that
	\begin{equation}
	\begin{aligned}
	g(\Sm \Tm^\herm )&= M + \Vert \Sm \Dm \Tm^\herm\Vert_{\sf F}^2 - 2\Re\{ \langle\Sm \Dm \Tm^\herm , \Sm  \Tm^\herm \rangle \}\\
	& = M + \Vert \Sm \Dm \Tm^\herm\Vert_{\sf F}^2 - 2\text{trace}(\Dm)\\
	&= M + \Vert\Vm \Vert_{\sf F}^2 - 2\sum_i \dv_i,.
	\end{aligned}
	\end{equation}
This completes the proof.
	\end{proof} 
This lemma provides an explicit formula for the orthogonal projection of a matrix into the unitary set as
\begin{equation}\label{eq:orthogonal_proj}
\Pc_{\Uc}:\bC^{M\times M}\to \Uc,\, \Vm = \Sm \Dm \Tm^\herm\to \Sm \Tm^\herm.
\end{equation}
The following theorem presents the main result of this work.
\begin{theorem}\label{th:convergance}
	Let $\Um^{(0)}\in \Uc$ be an initial point and consider the gradient projection update rule 
	\begin{equation}\label{eq:grad_proj}
	\Um^{(t+1)} = \Pc_{\Uc} \left( \Um^{(t)} - \alpha_t \nabla f (\Um^{(t)})  \right),\,\, t=0,1,\ldots,
	\end{equation}
	 with $\alpha_t \in (0,\frac{1}{L})$ for all $t$, where $L$ is the Lipschitz constant of $\nabla f(\Um)$. Then the sequence $\{ \Um^{(t)},\, t=0,1,\ldots \}$ converges to a stationary point of $f(\Um)$.
\end{theorem}


In order to prove Theorem \ref{th:convergance}, we need to first prove some useful properties of the ML optimization problem.

\subsection{Lipschitz Continuity of the Cost Gradient}
As a first step, we prove that the cost gradient $\nabla f(\Um)$ is Lipschitz continuous over $\Uc$. Note that the manifold $\Uc$ is a subset of the closed convex ball $\Bc$ ($\Uc \subset \Bc$) where 
$\Bc = \{ \Um\, :\, \Vert \Um \Vert_{\sf F}\le \sqrt{M} \}.$
One can show that $f(\Um)$ has Lipschitz continuous gradient over $\Bc$, i.e. there exists a constant $L$, such that
$\Vert \nabla f(\Um )- \nabla f(\Um') \Vert_{\sf F} \le L \Vert \Um - \Um' \Vert_{\sf F},$
for all $\Um,\Um' \in \Bc$. One way to prove this is by showing that the Hessian of $f(\Um)$ has bounded operator norm over $\Bc$. 
 Define the complex Hessian as the $M^2\times M^2$ square matrix $\nabla^2 f(\Um)$ whose elements are given as \cite{gunning2009analytic}
$
 [\nabla^2 f(\Um)]_{m,n} = \frac{\partial^2\, f(\Um)}{\partial [\vec (\Um)]_m  \partial [\vec (\Um)]_n^\ast },$
for $m,n\in [M^2]$, where 
$\vec (\Um) = [\uv_1^\transp,\ldots,\uv_M^\transp]^\transp$ is the vectorized version of $\Um$. Simple calculations show that the Hessian is a block-diagonal matrix with its $m$-th diagonal block given as
\begin{align}\label{eq:Hess_f_0}
\Dm_f^{(m)} = \sum_{k=1}^K
\frac{\widehat{\Sigmam}_k^\transp}{\uv_m^\herm \widehat{\Sigmam}_k \uv_m} - \sum_{k=1}^K \frac{\left( \widehat{\Sigmam}_k \uv_m \uv_m^\herm \widehat{\Sigmam}_k \right)^\transp}{\left(\uv_m^\herm \widehat{\Sigmam}_k \uv_m\right)^2},
\end{align}
\noindent so that we have $\nabla^2 f(\Um) = \text{blkdiag}\left( \Dm_f^{(1)},\ldots,\Dm_f^{(M)} \right)$. Note that both terms on the right-hand-side of \eqref{eq:Hess_f_0} are PSD and therefore $\Dm_f^{(m)}$ is the difference of two PSD matrices.
\begin{lemma}
The Hessian matrix $\nabla^2 f(\Um)$ is bounded in operator norm.
\end{lemma}
\begin{proof}
	Define the operator norm of a matrix $\Am\in \bC^{M^2\times M^2}$ as $\Vert \Am \Vert_{op} = \underset{\Vert\xv \Vert=1}{\sup} \frac{\Vert \Am \xv \Vert}{\Vert\xv \Vert}$, where $\Vert \cdot \Vert$ is the $\ell_2$ norm.
	For a block-diagonal matrix such as $\nabla^2 f (\Um) $, the operator norm is equal to the maximum of the operator norms of each individual block, i.e. $
	\Vert \nabla^2 f (\Um) \Vert_{op} = \underset{m}{\max}~ \Vert \Dm_f^{(m)} \Vert_{op}$.
	Using \eqref{eq:Hess_f_0}, the operator norm of block $m$ is bounded as
\scalebox{0.97}{$\Vert \Dm_f^{(m)} \Vert_{op} \le \max \left\{ \sum_k \frac{\Vert \widehat{\Sigmam}_k \Vert_{op}}{\uv_m^\herm \widehat{\Sigmam}_k \uv_m}\,,\, \sum_k\frac{\Vert \widehat{\Sigmam}_k\uv_m \Vert^2}{(\uv_m^\herm \widehat{\Sigmam}_k \uv_m)^2} \right\}$}
	where we used the fact that $\Dm_f^{(m)}$ is the difference of two PSD matrices and therefore its operator norm is bounded by the maximum of the operator norms of the two. Also, since the matrix $\widehat{\Sigmam}_k \uv_m \uv_m^\herm \widehat{\Sigmam}_k$ is of rank one, its operator norm is equal to $\Vert \widehat{\Sigmam}_k\uv_m \Vert^2$. Finally, since sample covariances are assumed to be non-singular, both arguments in $\max\{\cdot\}$ are finite. Taking the maximum over all $M$ bounds also results in a finite value and the proof is complete.
\end{proof}

Next we show that the Lipschitz constant of $\nabla f (\Um)$ is related to the operator norm of $\nabla^2 f(\Um)$.
\begin{lemma}
For a twice differentiable function $f(\Um)$ with Hessian bounded in operator norm as $\Vert \nabla^2 f(\Um) \Vert_{op} \le L$ for all $\Um$, the gradient $\nabla f (\Um)$ is Lipschitz continuous with Lipschitz constant $L$. 
\end{lemma}
\begin{proof}
We show that the Lipschitz continuity condition $\Vert  \nabla f(\Um)-\nabla f(\Um') \Vert_{\sf F}\le L \Vert \Um -\Um' \Vert_{\sf F} $ holds for any $\Um,\, \Um'$ via the following sequence of inequalities:
\begin{equation}\label{eq:Ineq_5}
\begin{aligned}
\Vert  \nabla f(\Um)-\nabla f(\Um') \Vert_{\sf F} 
&\overset{(a)}{\le} \underset{\Vert \Bm \Vert_{\sf F}=1}{\sup} \left| \left\langle \Bm , \nabla f(\Um)-\nabla f(\Um') \right\rangle \right|\\
& \hspace{-32mm}\scalebox{0.89}{$=\underset{\Vert \Bm \Vert_{\sf F}=1}{\sup} \left| \left\langle \vec(\Bm) ,\int_{0}^1 \nabla^2 f \left(\Um'+t(\Um - \Um')\right)\text{vec}(\Um - \Um')\, dt \right\rangle \right|$} 
\\
& \hspace{-32mm}\scalebox{0.9}{$\le \underset{\Vert \Bm \Vert_{\sf F}=1}{\sup}\int_{0}^1 \left| \left\langle \vec(\Bm) , \nabla^2 f \left(\Um'+t(\Um - \Um')\right)\text{vec}(\Um - \Um')\,  \right\rangle \right|dt$}\\
&
\hspace{-32mm}\scalebox{0.85}{$\overset{(b)}{\le}\underset{\Vert \Bm \Vert_{\sf F}=1}{\sup}\underset{t\in [0,1]}{\sup} \Vert \nabla^2 f \left(\Um'+t(\Um - \Um')\right) \Vert_{op} \Vert \text{vec}(\Um - \Um') \Vert \Vert  \vec(\Bm) \Vert$}\\
&
\hspace{-32mm}\scalebox{0.81}{$=\underset{t\in [0,1]}{\sup} \Vert \nabla^2 f \left(\Um'+t(\Um - \Um')\right) \Vert_{op} \Vert \text{vec}(\Um - \Um') \Vert $}\overset{(c)}{\le } L \Vert \Um - \Um' \Vert_{\sf F}. \\
\end{aligned}
\end{equation}
Inequality $(a)$ holds by taking into account the fact that for the particular value of $\Bm$ as $\Bm_0 = \frac{\nabla f(\Um)-\nabla f(\Um') }{\Vert  \nabla f(\Um)-\nabla f(\Um') \Vert_{\sf F} }$ we have
$ \langle \Bm_0 , \nabla f(\Um)-\nabla f(\Um') \rangle = \Vert  \nabla f(\Um)-\nabla f(\Um') \Vert_{\sf F}.$ Inequality $(b)$ comes from an application of the Cauchy-Schwarz inequality and the definition of the operator norm. Finally, inequality $(c)$ holds due to the assumption on the boundedness of the Hessian operator norm, i.e. $\Vert\nabla^2 f(\Um)\Vert_{op}\le L$ for all $\Um$, and the proof is complete.
\end{proof}
The next lemma emerges as a consequence of the discussion above.
\begin{lemma}
	For any pair of matrices $\Um,\Um'\in \Bc$ we have
	\begin{equation}\label{eq:descent}
	f(\Um)\le f(\Um') + \langle \nabla f(\Um'),\Um - \Um'\rangle + \frac{L}{2}\Vert\Um-\Um'  \Vert_{\sf F}^2,
	\end{equation}
	where $L$ is the gradient Lipschitz constant.
	\end{lemma} 
\begin{proof}
	See \cite{bertsekas2015convex}, proposition 6.1.2. 
\end{proof}
This lemma is used as a tool to prove the convergence of PGD to a stationary point, as outlined by Theorem \ref{th:convergance}. 
\subsection{\textbf{Proof of Theorem \ref{th:convergance}}}
We start by replacing $\Um'$ with $\Um^{(t)}$ in \eqref{eq:descent} and defining the RHS of \eqref{eq:descent} as the proxy function
\begin{equation}\label{eq:proxy}
\scalebox{0.92}{$\fproxy^{(t)} (\Um)= f(\Um^{(t)}) + \langle \nabla f(\Um^{(t)}),\Um - \Um^{(t)}\rangle + \frac{L}{2}\Vert\Um-\Um^{(t)}  \Vert_{\sf F}^2,$}
\end{equation} 
at point $\Um^{(t)}$, such that we have
\begin{equation}\label{eq:proxy_bound}
f(\Um)\le \fproxy^{(t)} (\Um) 
\end{equation}
for all $\Um\in \Uc$ and 
$
f(\Um^{(t)})= \fproxy^{(t)} (\Um^{(t)}).
$
Now let us show that the point $\Um^{(t+1)}=\Pc_{\Uc} \left( \Um^{(t)} - \alpha_t \nabla f (\Um^{(t)})  \right)$ is indeed a minimizer of $\fproxy^{(t)}(\Um)$ over $\Uc$ with $\alpha_t=\frac{1}{L}$. To see this, note that we can expand $\fproxy^{(t)}(\Um)$ as
$ \fproxy^{(t)}(\Um) = \langle \nabla f(\Um^{(t)}),\Um\rangle-L \langle \Um^{(t)},\Um \rangle + {\sf const.}, $
for all unitary $\Um$. Then, minimizing $\fproxy^{(t)}(\Um) $ is equivalent to the maximization problem:
$\underset{\Um^\herm \Um = \mathbf{I}_{M}}{\text{maximize}}\,\, \langle \Um^{(t)} -\frac{1}{L}\nabla f(\Um^{(t)}),\Um \rangle$.
But the maximum of this objective is achieved at the point $\Um^\star = \Sm_t \Tm_t^\herm$, where $\Sm_t$ and $\Tm_t$ are matrices of left and right eigenvectors in the SVD form $ \Um^{(t)} -\frac{1}{L}\nabla f(\Um^{(t)}) = \Sm_t\Dm_t \Tm_t^\herm$. But this implies that $\Um^\star =  \Sm_t \Tm_t^\herm = \Pc_{\Uc} \left( \Um^{(t)} - \frac{1}{L}\nabla f (\Um^{(t)})  \right)=\Um^{(t+1)}$ and $\Um^{(t+1)}$ is a minimizer of $\fproxy^{(t)}(\Um)$. The chain of inequalities below immediately follows:
\begin{equation}\label{eq:ineqs}
\scalebox{0.92}{$f(\Um^{(t+1)})\overset{(a)}{\le } \fproxy^{(t)}(\Um^{(t+1)})\overset{(b)}{\le } \fproxy^{(t)}(\Um^{(t)})\overset{(c)}{= } f(\Um^{(t)}),$}
\end{equation}
where $(a)$ follows from \eqref{eq:proxy_bound}, $(b)$ follows from the fact that $\Um^{(t+1)}$ is aminimizer of $\fproxy^{(t)}(\Um)$, and $(c)$ is a result of $f(\Um^{(t)})= \fproxy^{(t)} (\Um^{(t)})$. Therefore we have
$f(\Um^{(t+1)}) \le f(\Um^{(t)}),
$
for $t=0,1,\ldots$ and since $f(\Um)$ is bounded from below, the gradient projection sequence $\{ \Um^{(t)},\, t=0,1,\ldots \}$  converges to a stationary point of $f(\Um)$. $\hfill \QED$
\begin{figure}[t]
	\centering
	\includegraphics[ width=0.37\textwidth]{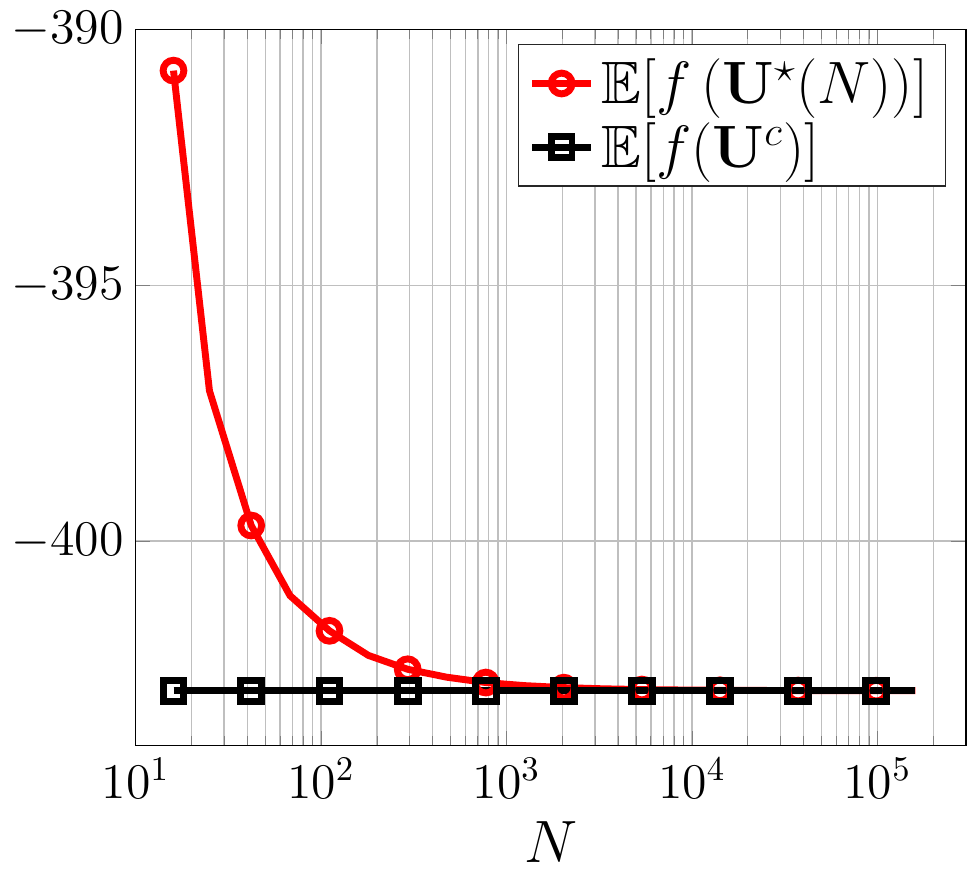}
	\caption{The average cost $f(\Um)$ as a function of the number of samples, for the solution of PGD $\Um^\star (N)$ and the global optimum $\Um^c$. We have set $M=16$ and $K=8$. }
	\label{fig:cost_global}
\end{figure} 
\section{Simulation Results}
To study the performance of our proposed method, in this section we provide some empirical results. 
\subsection{Jointly Diagonalizable Covariances}
The case of jointly diagonalizable covariances is especially interesting, as we know by Theorem \ref{thm_global_opt} that the global optimum of the ML problem is given by the shared CES $\Um^c$ (see \eqref{eq:commut_cov}). In order to assess the performance of our method, we compare the ML cost as a function of the number of random samples per process $N$, to the cost at the global minimum. Consider a signal dimension (number of antennas) $M=16$ and a number of processes (number of users) $K=8$. We generate a random unitary matrix as the CES $\Um^c$ by calculating the eigenvectors of a random matrix of size $M\times M$ with i.i.d complex Gaussian elements. Also, for each process $k\in [K]$, we generate a random vector of eigenvalues $\lambdav_k$ with i.i.d, positive elements given as $\lambdav_{k,m}=|\rho_m|$ where $\rho_m \sim \Nc (0,1)$. Then we form the covariance matrix of user $k$ as $\widetilde{\Sigmam}=\Um^c \text{diag}(\lambdav_k)\Um^{c\, \herm}$. We also normalize the covariances to have trace equal to one. This way we have randomly generated covariances with a shared CES. Now, having the covariances, we can generate random realizations for each process for different sample sizes $N$. We run a Monte Carlo simulation with $1000$ iterations, at each iteration generating covariances as stated above, then for each sample size $N$ we run our proposed PGD method which converges to a point $\Um^\star (N)$ (explicitly noting the dependence on $N$). Then, we average the cost function $\Um^\star (N)$ in \eqref{eq:ML_formula_2} over the Monte Carlo iterations for each value of $N$ and compare it to the average cost at the global optimum $f(\Um^c )$. Note that the latter of course is not a function of $N$. The PGD method is initialized with a random unitary matrix. Theorem \ref{th:convergance} guarantees convergence to a stationary point when the step size is chosen as $\alpha_t\in (0,1/L)$. However, practically we can be more ambitious by choosing larger step sizes as $\alpha_t = \frac{\alpha_0}{t},~t=1,2,\ldots$ with $\alpha_0 =2$ and our simulation results show that even with this choice, PGD converges.

Fig. \ref{fig:cost_global} illustrates the result. We denote the solution of our method with $\Um^\star (N)$, to explicitly highlight its dependence on the number of samples $N$. The interesting fact about this result is that, as the number of samples gets larger, the PGD with random initialization always converges to the global optimum, as its cost value is the same as that in the optimal point $\Um^c$. This is an empirical evidence for the convergence of our proposed method to the global solution for jointly diagonalizable covariances. 

\begin{figure}[t]
	\centering
	\includegraphics[ width=0.386\textwidth]{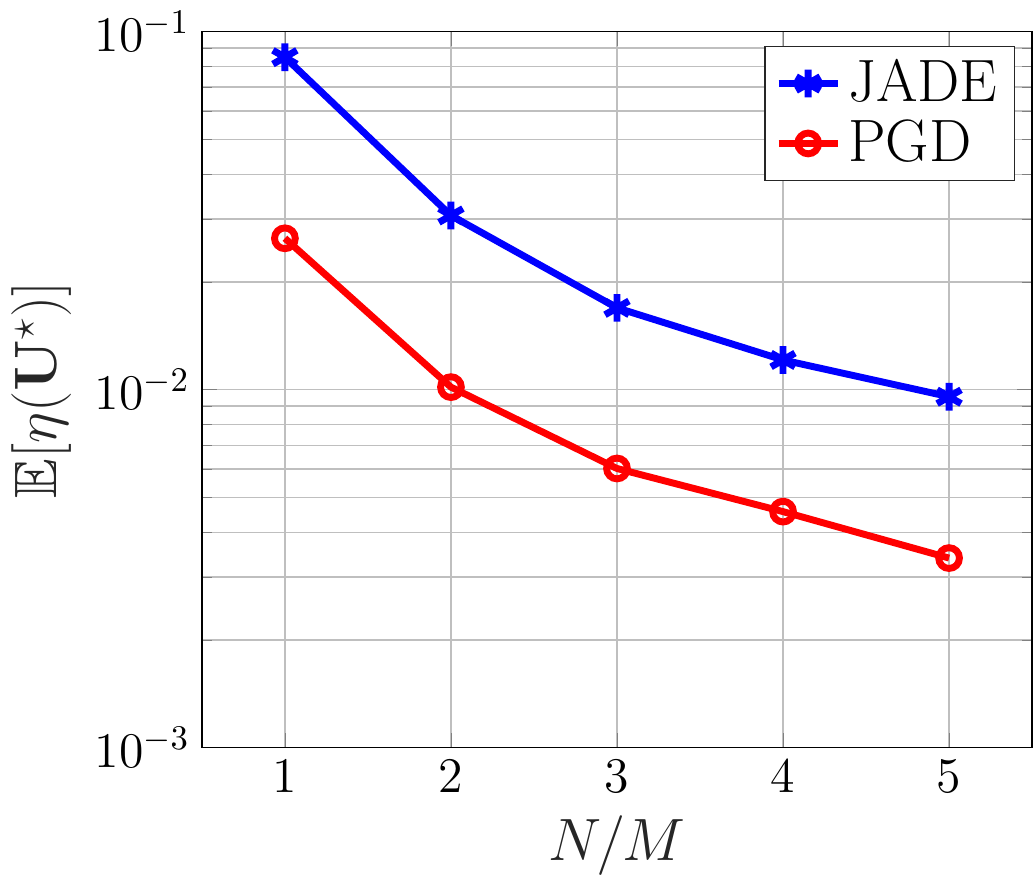}
	\caption{The average diagonalization metric as function of the sampling ratio $N/M$, for the solution of our proposed PGD method vs the JADE method. We set $M=16$ and $K=8$. }
	\label{fig:eta_func}
\end{figure} 
\subsection{Non-Jointly Diagonalizable Covariances}
As a different scenario, we
consider covariances that are not jointly diagonalizable. This is done by generating a different random unitary eigenvector matrix for each process separately as $\widetilde{\Sigmam}=\Um_k \text{diag}(\lambdav_k)\Um_k^{ \herm}$. The eigenvalue vectors $\lambdav_k$ are generated as before and we normalize the covariances to have unit trace. So, in this case we do not have a CES and the PGD method yields a unitary matrix that approximately jointly diagonalizes the covariances. Since we do not have the global optimum in this case, we compare our method to the JADE algorithm, which is a classic Jacobian-based method for joint covariance diagonalization (we do not explain the details of this method here due to space limitations and refer the reader to \cite{cardoso1993blind} for a full account). 
 
One way to measure the performance of joint diagonalization methods is by defining the metric $\eta : \bC^{M\times M}\to [0,1]$:
\begin{equation}\label{eq:eta}
\eta (\Um) = 1-\frac{1}{K}\sum_{k=1}^K\frac{\Vert \text{diag}(\Um^\herm \Sigmam_k\Um) \Vert}{\Vert \Sigmam_k \Vert_{\sf F}},
\end{equation}
where $\text{diag}(\cdot)$ with a matrix argument as in \eqref{eq:eta} denotes the $M$-dim vector of the diagonal elements of its argument. The smaller the value of $\eta (\Um)$ is, the better $\Um$ jointly diagonalizes the covariances. In the extreme case, If $\Um$ diagonalizes all covariance matrices, we have $\eta (\Um)=0$.

The joint diagonalization metric is empirically averaged over $1000$ Monte Carlo simulations and for different sample sizes for the solutions of our proposed PGD method and the JADE method. Fig. \ref{fig:eta_func} illustrates the results. It clearly shows that for the ranges of sample sizes considered here, the proposed PGD method outperforms the classic JADE method, yielding smaller values of the diagonalization metric on average, and hence achieving a better joint diagonalization of the covariances. 

\subsection{CES for ULA: PGD vs the Fourier Basis}
For a ULA it is usually taken for granted that the CES is given by the Fourier basis vectors. While this is true in an asymptotic sense thanks to the Szeg{\"o} theorem, it does not hold for small to moderate array sizes. We conclude our simulations by showing that, in fact the unitary basis produced by our proposed method better diagonalizes ULA covariances, compared to the DFT basis. 

\begin{figure}[t]
	\centering
	\includegraphics[ width=0.4\textwidth]{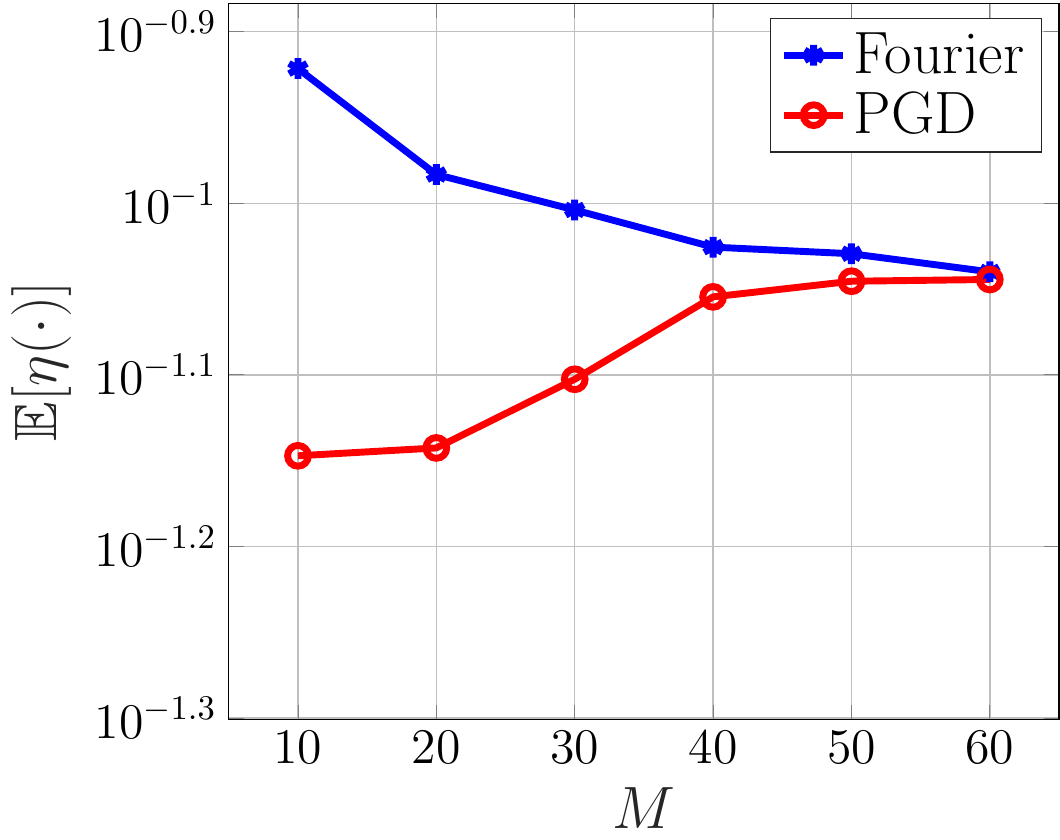}
	\caption{The average diagonalization metric as function of the number of antennas $M$, for the solution of our proposed PGD method vs the Fourier basis. Here we have set $K=5$. }
	\label{fig:U_vs_F}
\end{figure} 

We consider $K=5$ $M$-dimensional ULA covariances, generated randomly according to \eqref{eq:ch_cov_0} (see \cite{khalilsarai2018fdd} on randomly generating ULA covariances). Since each covariance is associated with a user that occupies a limited angular range as seen from the BS, we generate the covariances such that each of them has an ``effective" rank of $r_k=\text{effrank}(\widetilde{\Sigmam}_k)=\lceil\frac{M}{2} \rceil$. The effective rank is equivalent to the number of significant covariance eigenvalues as well as channel angular sparsity. 

We plot the expected joint diagonalization metric $\bE [ \eta (\cdot) ]$ as a function of the number of antennas for the unitary matrix yielded by our method as well as for the Fourier basis $\Fm$ where $[\Fm]_{m,n}=\frac{1}{\sqrt{M}}e^{j2\pi \frac{(m-1)(n-1)}{M}},\, m,n\in [M]$. We assume that the sample covariance has converged, i.e. $\widehat{ \Sigmam}_k = \widehat{ \Sigmam}_k $ for all $k$. The expectation $\bE [ \eta (\cdot) ]$ is taken over random covariance realizations and is calculated empirically over 100 Monte-Carlo loops. Fig. \ref{fig:U_vs_F} illustrates the result. As we can see, the basis given by the PGD method achieves better diagonalization (smaller $\eta$ values) than the Fourier basis, which shows that using our method is in fact preferable even for the diagonalization of ULA covariances. Also, as $M$ increases, the two bases have closer diagonalization metrics since we are approaching the asymptotic regime in which the Fourier basis approximately diagonalizes the Toeplitz ULA covariances. 

\section{Conclusion}
We presented a framework for estimating the common eigenstructure for a set of covariance matrices. Our approach was based on maximizing the likelihood function of the postulated common eigenvectors set, given random vector realizations. We proposed the PGD method and proved its convergence to a stationary point of the likelihood function. Our empirical results illustrated that the proposed method converges to the global optimum when the covariances indeed share a common eigenvectors set. It outperforms the classic JADE method in all ranges of sample sizes. It also performs better than the Fourier basis in diagonalizing covariances generated by a ULA geometry in a multi-user MIMO setup.

	{\small
		\bibliographystyle{IEEEtran}
		\bibliography{references}

\begin{thebibliography}{10}
\providecommand{\url}[1]{#1}
\csname url@samestyle\endcsname
\providecommand{\newblock}{\relax}
\providecommand{\bibinfo}[2]{#2}
\providecommand{\BIBentrySTDinterwordspacing}{\spaceskip=0pt\relax}
\providecommand{\BIBentryALTinterwordstretchfactor}{4}
\providecommand{\BIBentryALTinterwordspacing}{\spaceskip=\fontdimen2\font plus
\BIBentryALTinterwordstretchfactor\fontdimen3\font minus
  \fontdimen4\font\relax}
\providecommand{\BIBforeignlanguage}[2]{{%
\expandafter\ifx\csname l@#1\endcsname\relax
\typeout{** WARNING: IEEEtran.bst: No hyphenation pattern has been}%
\typeout{** loaded for the language `#1'. Using the pattern for}%
\typeout{** the default language instead.}%
\else
\language=\csname l@#1\endcsname
\fi
#2}}
\providecommand{\BIBdecl}{\relax}
\BIBdecl

\bibitem{larsson2014massive}
E.~G. Larsson, O.~Edfors, F.~Tufvesson, and T.~L. Marzetta, ``Massive {MIMO}
  for next generation wireless systems,'' \emph{IEEE communications magazine},
  vol.~52, no.~2, pp. 186--195, 2014.

\bibitem{bjornson2016massive}
E.~Bj{\"o}rnson, E.~G. Larsson, and T.~L. Marzetta, ``Massive {MIMO}: {T}en
  myths and one critical question,'' \emph{IEEE Communications Magazine},
  vol.~54, no.~2, pp. 114--123, 2016.

\bibitem{yang2013performance}
H.~Yang and T.~L. Marzetta, ``Performance of conjugate and zero-forcing
  beamforming in large-scale antenna systems,'' \emph{IEEE Journal on Selected
  Areas in Communications}, vol.~31, no.~2, pp. 172--179, 2013.

\bibitem{heath2016overview}
R.~W. Heath, N.~Gonzalez-Prelcic, S.~Rangan, W.~Roh, and A.~M. Sayeed, ``An
  overview of signal processing techniques for millimeter wave {MIMO}
  systems,'' \emph{IEEE journal of selected topics in signal processing},
  vol.~10, no.~3, pp. 436--453, 2016.

\bibitem{marzetta2010noncooperative}
T.~L. Marzetta, ``Noncooperative cellular wireless with unlimited numbers of
  base station antennas,'' \emph{IEEE transactions on wireless communications},
  vol.~9, no.~11, pp. 3590--3600, 2010.

\bibitem{shepard2012argos}
C.~Shepard, H.~Yu, N.~Anand, E.~Li, T.~Marzetta, R.~Yang, and L.~Zhong,
  ``{A}rgos: Practical many-antenna base stations,'' in \emph{Proceedings of
  the 18th annual international conference on Mobile computing and networking},
  2012, pp. 53--64.

\bibitem{adhikary2013joint}
A.~Adhikary, J.~Nam, J.-Y. Ahn, and G.~Caire, ``Joint spatial division and
  multiplexing—the large-scale array regime,'' \emph{IEEE transactions on
  information theory}, vol.~59, no.~10, pp. 6441--6463, 2013.

\bibitem{haghighatshoar2016massive}
S.~Haghighatshoar and G.~Caire, ``Massive {MIMO} channel subspace estimation
  from low-dimensional projections,'' \emph{IEEE Transactions on Signal
  Processing}, vol.~65, no.~2, pp. 303--318, 2016.

\bibitem{haghighatshoar2018low}
------, ``Low-complexity massive {MIMO} subspace estimation and tracking from
  low-dimensional projections,'' \emph{IEEE Transactions on Signal Processing},
  vol.~66, no.~7, pp. 1832--1844, 2018.

\bibitem{khalilsarai2018fdd}
M.~B. Khalilsarai, S.~Haghighatshoar, X.~Yi, and G.~Caire, ``{FDD} massive
  {MIMO} via {UL}/{DL} channel covariance extrapolation and active channel
  sparsification,'' \emph{IEEE Transactions on Wireless Communications},
  vol.~18, no.~1, pp. 121--135, 2018.

\bibitem{haghighatshoar2017massive}
S.~Haghighatshoar and G.~Caire, ``Massive {MIMO} pilot decontamination and
  channel interpolation via wideband sparse channel estimation,'' \emph{IEEE
  Transactions on Wireless Communications}, vol.~16, no.~12, pp. 8316--8332,
  2017.

\bibitem{khalilsarai2018achieve}
M.~B. Khalilsarai, S.~Haghighatshoar, and G.~Caire, ``How to achieve massive
  {MIMO} gains in {FDD} systems?'' in \emph{2018 IEEE 19th International
  Workshop on Signal Processing Advances in Wireless Communications
  (SPAWC)}.\hskip 1em plus 0.5em minus 0.4em\relax IEEE, 2018, pp. 1--5.

\bibitem{nguyen2013compressive}
S.~L.~H. Nguyen and A.~Ghrayeb, ``Compressive sensing-based channel estimation
  for massive multiuser {MIMO} systems,'' in \emph{2013 IEEE Wireless
  Communications and Networking Conference (WCNC)}.\hskip 1em plus 0.5em minus
  0.4em\relax IEEE, 2013, pp. 2890--2895.

\bibitem{wunder2018hierarchical}
G.~Wunder, I.~Roth, M.~Barzegar, A.~Flinth, S.~Haghighatshoar, G.~Caire, and
  G.~Kutyniok, ``Hierarchical sparse channel estimation for massive {MIMO},''
  in \emph{WSA 2018; 22nd International ITG Workshop on Smart Antennas}.\hskip
  1em plus 0.5em minus 0.4em\relax VDE, 2018, pp. 1--8.

\bibitem{unser2014introduction}
M.~Unser and P.~D. Tafti, \emph{An introduction to sparse stochastic
  processes}.\hskip 1em plus 0.5em minus 0.4em\relax Cambridge University
  Press, 2014.

\bibitem{sayeed2013beamspace}
A.~Sayeed and J.~Brady, ``Beamspace {MIMO} for high-dimensional multiuser
  communication at millimeter-wave frequencies,'' in \emph{2013 IEEE global
  communications conference (GLOBECOM)}.\hskip 1em plus 0.5em minus 0.4em\relax
  IEEE, 2013, pp. 3679--3684.

\bibitem{cardoso1993blind}
J.-F. Cardoso and A.~Souloumiac, ``Blind beamforming for non-gaussian
  signals,'' in \emph{IEE proceedings F (radar and signal processing)}, vol.
  140, no.~6.\hskip 1em plus 0.5em minus 0.4em\relax IET, 1993, pp. 362--370.

\bibitem{peajcariaac1992convex}
J.~E. Peajcariaac and Y.~L. Tong, \emph{Convex functions, partial orderings,
  and statistical applications}.\hskip 1em plus 0.5em minus 0.4em\relax
  Academic Press, 1992.

\bibitem{bertsekas2015convex}
D.~P. Bertsekas and A.~Scientific, \emph{Convex optimization algorithms}.\hskip
  1em plus 0.5em minus 0.4em\relax Athena Scientific Belmont, 2015.

\bibitem{manton2002optimization}
J.~H. Manton, ``Optimization algorithms exploiting unitary constraints,''
  \emph{IEEE Transactions on Signal Processing}, vol.~50, no.~3, pp. 635--650,
  2002.

\bibitem{mirsky1975trace}
L.~Mirsky, ``A trace inequality of {J}ohn von {N}eumann,'' \emph{Monatshefte
  f{\"u}r mathematik}, vol.~79, no.~4, pp. 303--306, 1975.

\bibitem{gunning2009analytic}
R.~C. Gunning and H.~Rossi, \emph{Analytic functions of several complex
  variables}.\hskip 1em plus 0.5em minus 0.4em\relax American Mathematical
  Soc., 2009, vol. 368.

\end{thebibliography}
	}
	
\end{document}